\documentclass[a4paper,12pt]{amsart}
\pdfoutput=1

\usepackage[sc]{mathpazo}
\usepackage[T1]{fontenc}
\usepackage{subfig}
\usepackage{graphicx}
\usepackage{amsmath}
\usepackage{amssymb}
\usepackage[active]{srcltx}
\usepackage{hyperref}

\newtheorem{thm}{Theorem}[section]%
\newtheorem{lem}{Lemma}[section]%
\newtheorem{prop}{Proposition}[section]%
\newtheorem{cor}{Corollary}[section]%
\theoremstyle{definition}
\newtheorem{defn}{Definition}[section]
\theoremstyle{remark}
\theoremstyle{plain}

\def\CC{{\mathbb C}}

\def\NN{{\mathbb N}}
\def\QQ{{\mathbb Q}}

\def\RR{{\mathbb R}}
\def\TT{{\mathbb T}}

\def\ZZ{{\mathbb Z}}

\def\scrD{{\mathcal D}}

\def\scrN{{\mathcal N}}

\def\scrS{{\mathcal S}}

\def\e{\mathrm{e}}
\def\i{\mathrm{i}}

\def\L{\operatorname{L{}}}

\def\Op{\operatorname{Op}}

\def\vol{\operatorname{vol}}

\newcommand\blfootnote[1]{%
  \begingroup
  \renewcommand\thefootnote{}\footnote{#1}%
  \addtocounter{footnote}{-1}%
  \endgroup
}

\title{On the Phase-Space Distribution of Bloch Eigenmodes for Periodic Point Scatterers}
\author{Jory Griffin}
\date{\today}

\begin{document}

\begin{abstract} Consider the 3-dimensional Laplacian with a potential described by point scatterers placed on the integer lattice. We prove that for Floquet-Bloch modes with fixed quasi-momentum satisfying a certain Diophantine condition, there is a subsequence of eigenvalues of positive density whose eigenfunctions exhibit equidistribution in position space and localisation in momentum space. This result complements the result of Uebersch\"{a}r and Kurlberg \cite{Ueberschaer_Kurlberg} who show momentum localisation for zero quasi-momentum in 2-dimensions and is the first result in this direction in 3-dimensions. 
\end{abstract}

\maketitle
\blfootnote{The research leading to these results has received funding from the European Research Council under the European Union's Seventh Framework Programme (FP/2007-2013) / ERC Grant Agreement n. 291147.}

\section{Introduction}
The phase space distribution of quantum eigenfunctions for large energies remains in general an unsolved problem - specifically, one would like to know whether the eigenfunctions of a given system exhibit equidistribution or some degree of localisation (or indeed both). We are motivated by the physical problem concerning propagation through a cubic crystal lattice of scatterers. It is well known that when considering a scattering problem in which the wavelength is much larger than the radius of the scatterer, we can replace the scattering potential with a Dirac $\delta$ point potential. This approach is perhaps most famously used in the Kronig-Penney model \cite{Kronig_Penney} which considers the one dimensional Schr\"{o}dinger equation with a Dirac comb potential. A thorough treatment of models of this type, as well as higher dimensional analogs can be found in \cite{solvable_models}. Periodic problems of this sort can be tackled with Floquet-Bloch theory which allows us to reduce a periodic problem in $\RR^d$ to a family of quasiperiodic problems on $\TT^d$ parametrised by their Bloch vector or quasimomentum $k \in \TT^d$. 

For zero quasimomentum the problem of limiting phase space distributions has been studied in two dimensions by Rudnick and Uebersch\"ar~\cite{Rudnick_Uebershaer}, and Uebersch\"ar and Kurlberg~\cite{Ueberschaer_Kurlberg,Ueberschaer_Kurlberg_QE}, who showed that almost all eigenfunctions equidistribute in position space for all tori, and that in momentum space almost all eigenfunctions \textit{either} equidistribute for square tori, or localise for tori with a diophantine ratio of side lengths. These results were partially generalised to three dimensions by Yesha~\cite{Yesha_EigenfunctionStatistics,Yesha_QE} who showed that for the cubic torus, \textit{all} eigenfunctions equidistribute in position space, and that almost all eigenfunctions equidistribute in phase space. These results are further complimented by Kurlberg and Rosenzweig \cite{Kurlberg_Rosenzweig} who show the existence of localisation in position representation in $2$ dimensions and momentum representation in both $2$ and $3$ dimensions.  In this paper we generalise the results on the cubic torus to include nonzero quasimomentum, which destroys the high eigenvalue multiplicity and consequently, the equidistribution observed by Yesha. The proof follows a similar blueprint, but relies heavily on a result concerning the convergence of the two-point correlation function for inhomogeneous quadratic forms \cite{Marklof_IQF2}.

\begin{figure}[h]
\centering
\subfloat[$\lambda_{100} \approx$ 100.03]{
  \includegraphics[width=60mm]{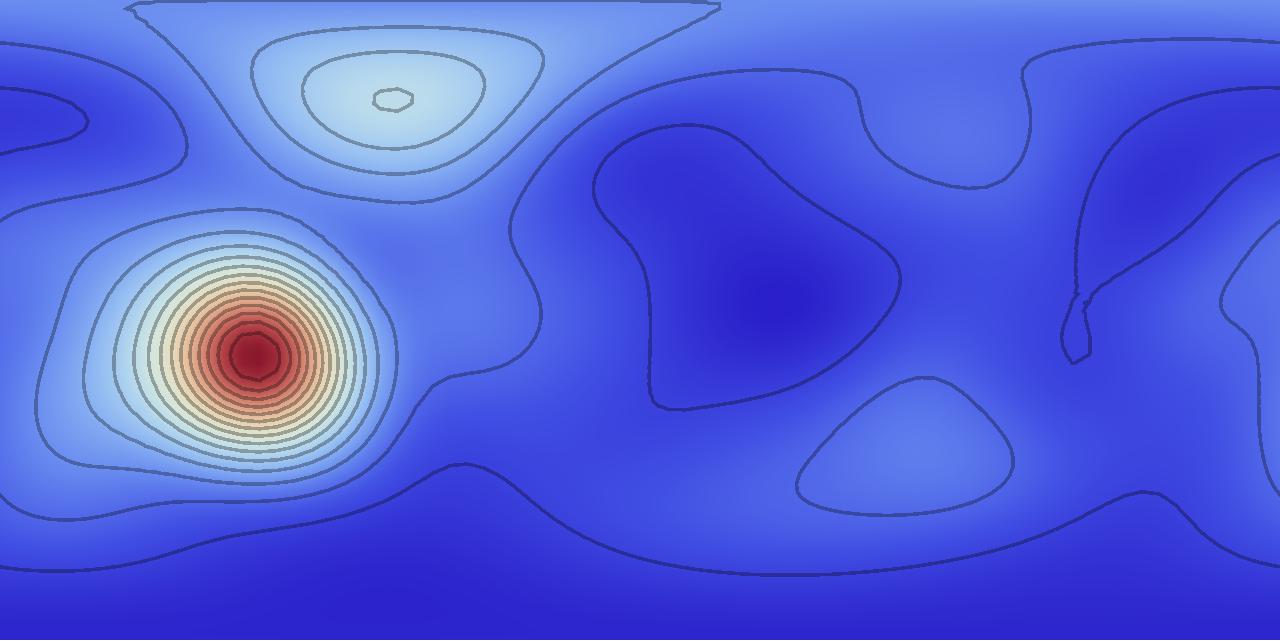}
}
\subfloat[$\lambda_{101} \approx$ 100.04]{
  \includegraphics[width=60mm]{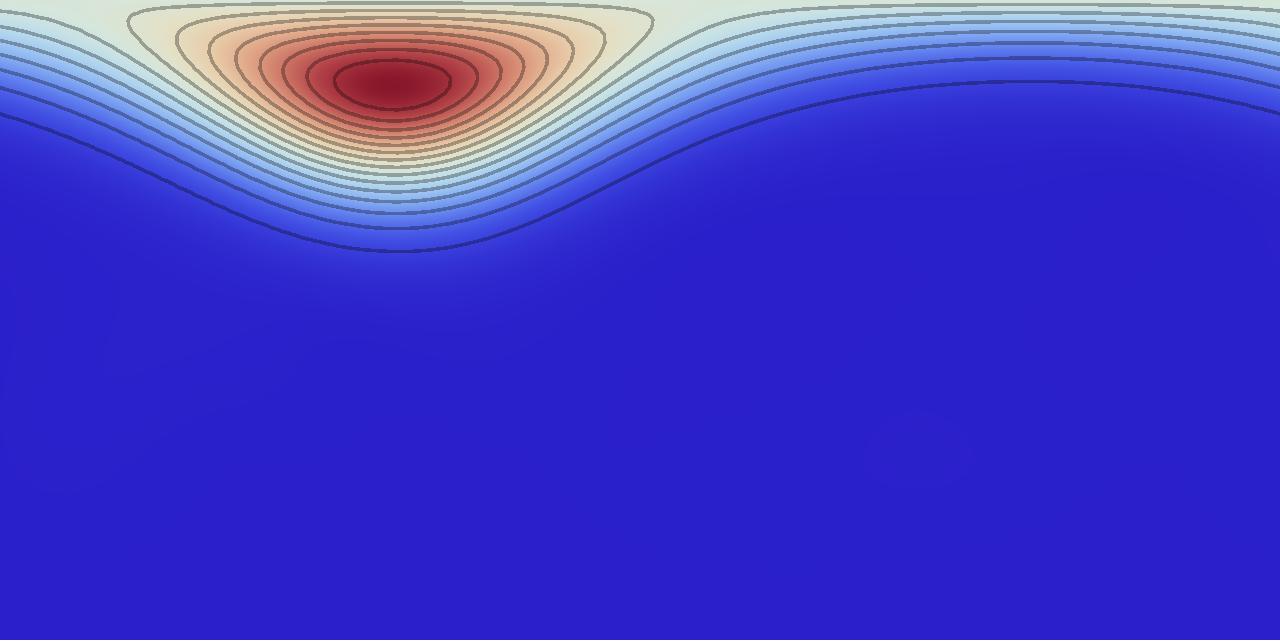}
}
\hspace{0mm}
\subfloat[$\lambda_{102} \approx$ 100.06]{
  \includegraphics[width=60mm]{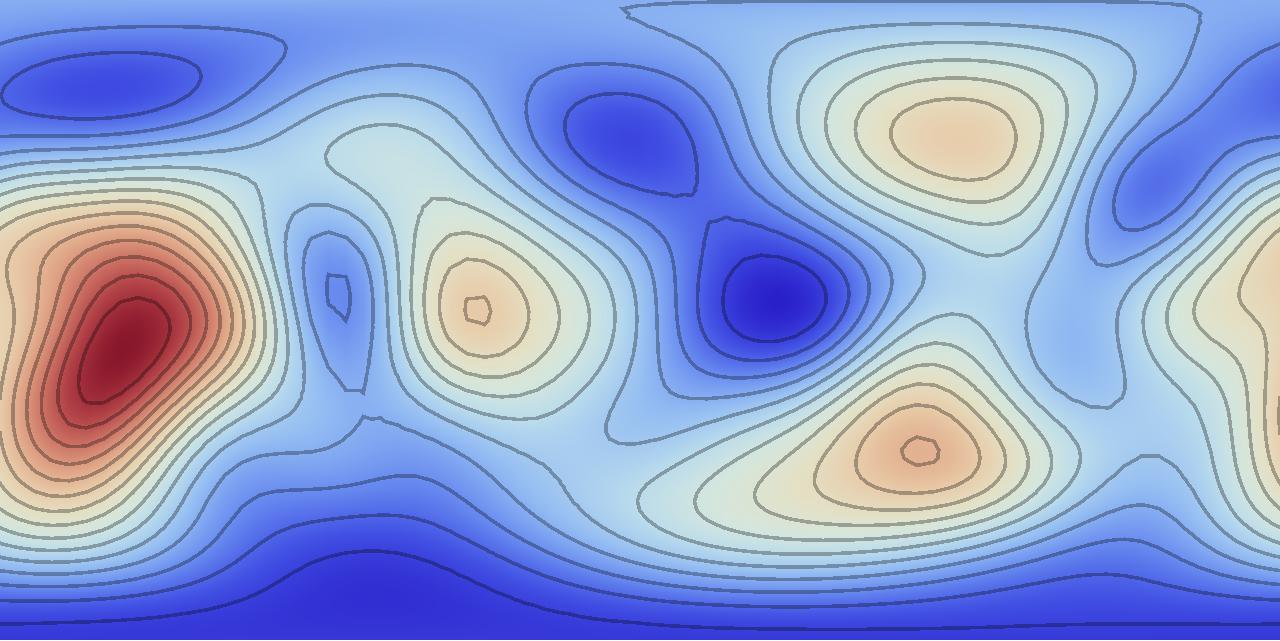}
}
\subfloat[$\lambda_{103} \approx$ 100.09]{
  \includegraphics[width=60mm]{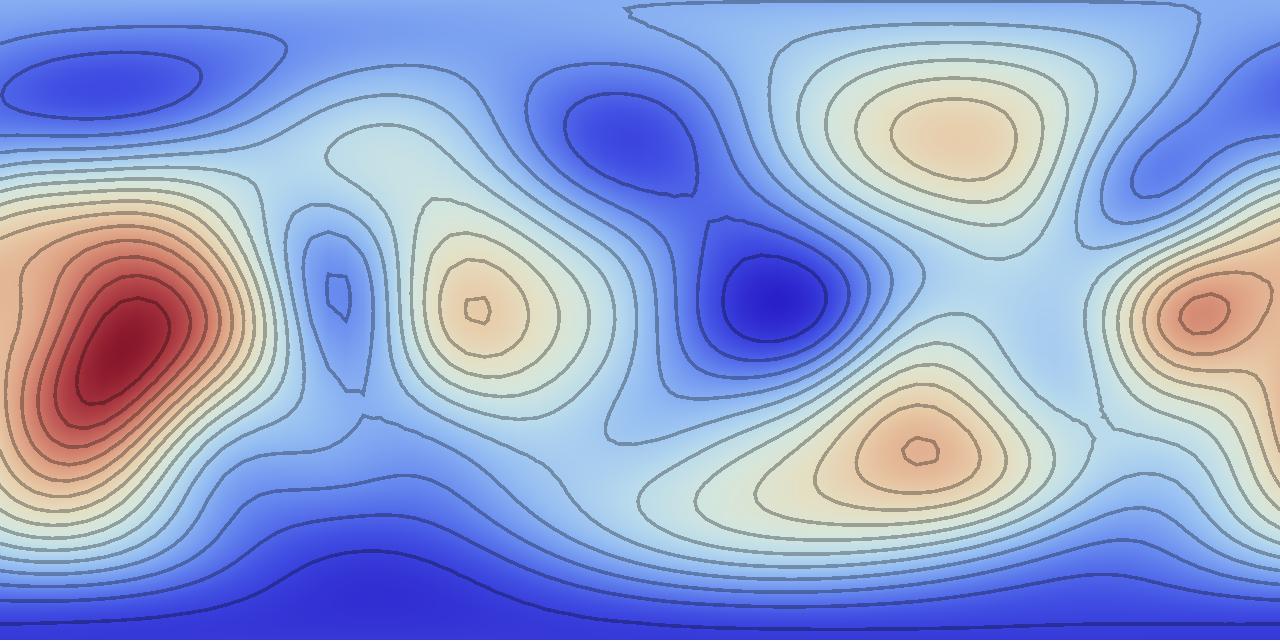}
}
\hspace{0mm}
\subfloat[$\lambda_{104} \approx$ 100.11]{
  \includegraphics[width=60mm]{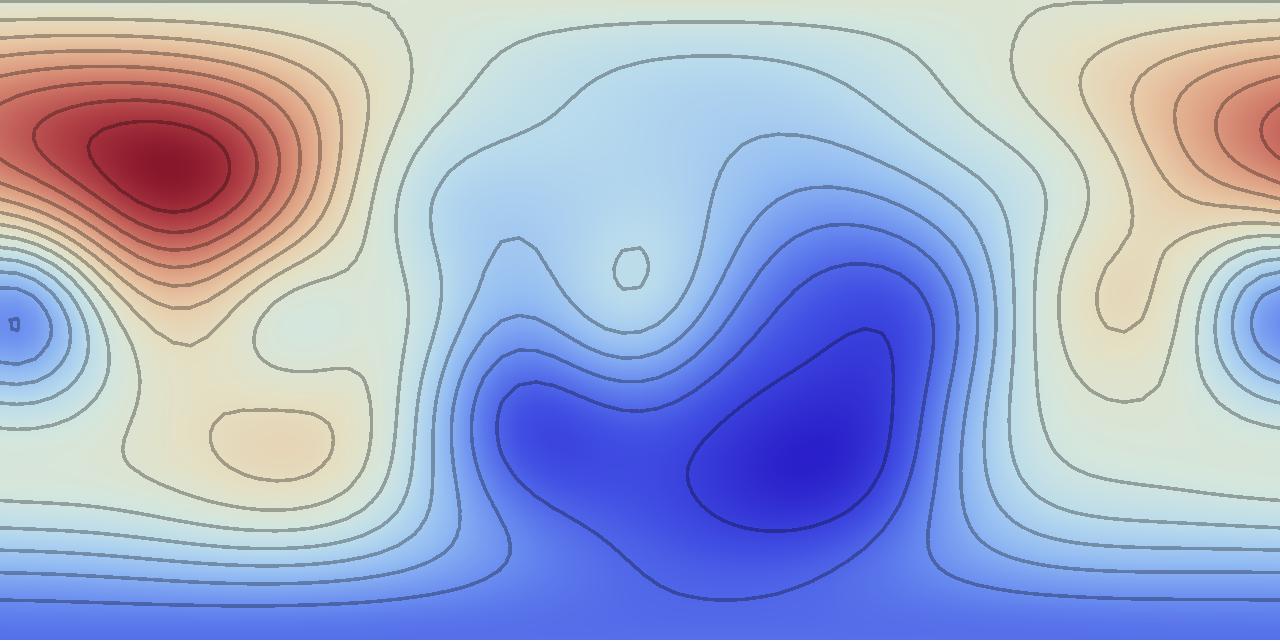}
}
\subfloat[$\lambda_{105} \approx$ 100.13]{
  \includegraphics[width=60mm]{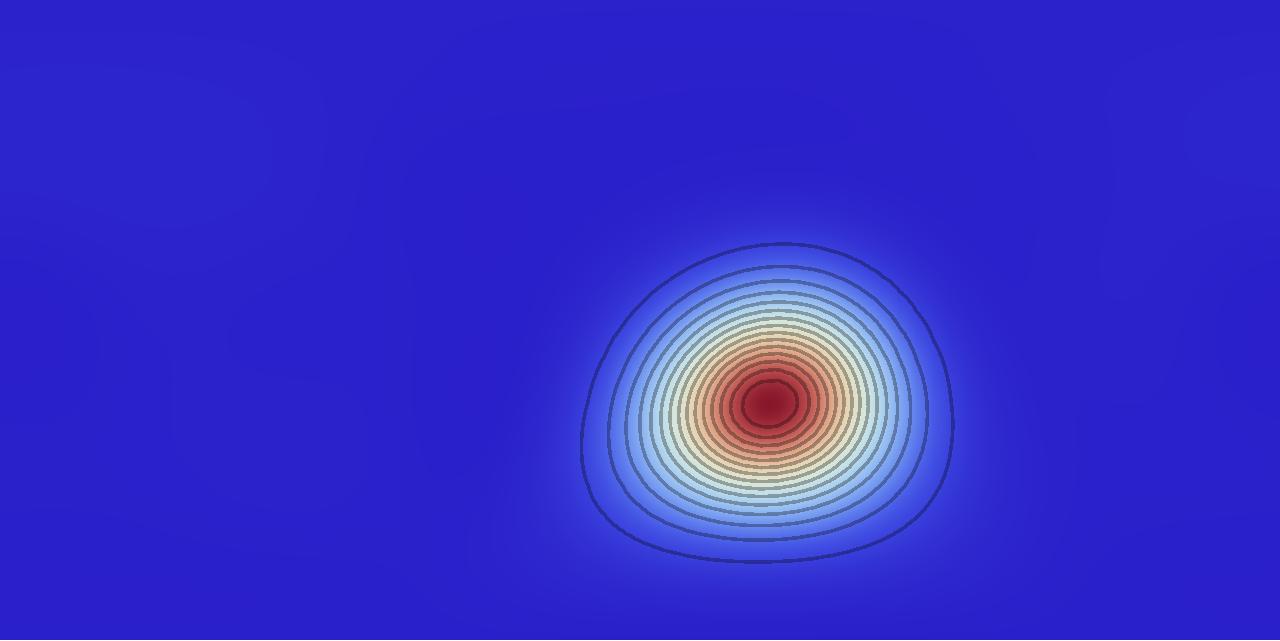}
}
\caption{Six consecutive eigenfunction density plots on the plane $(\theta,\phi)$ showing the distribution of momentum directions. We use fixed quasimomentum $k = (\tfrac{1}{\sqrt{2}}, \tfrac{1}{\sqrt{3}},\tfrac{1}{\sqrt{5}})$.}
\end{figure}

Problems of this type have been studied extensively in the Quantum Chaos literature since \v{S}eba~\cite{Seba} who considered a rectangular billiard with a point scatterer at some given point. The \v{S}eba billiard was constructed as an example of an intermediate system, meaning one that is classically integrable (the point scatterer affects only a zero measure set of trajectories) yet exhibits properties typical of chaotic systems \cite{Seba2,Shig1,Shig2,Shig3}. This is interesting in view of Shnirelman's theorem~\cite{Colin,Shnirelman,Zelditch} which states that for classically ergodic systems, a density one subsequence of eigenfunctions equidistributes in phase space, yet when the classical dynamics is integrable eigenfunctions tend to localise or scar. We are interested in the Schr\"{o}dinger equation on $\RR^3$ with potential described by point scatterers placed on $2\pi\ZZ^3$ which is described by the formal operator
\begin{align}
-\Delta + c \sum_{j \in 2\pi\ZZ^3} \delta_{j+x_0}.
\end{align}
This operator is unitarily equivalent via a gauge transformation to a direct integral over quasimomenta $k$. That is, we can instead consider a related quasiperiodic problem on the torus which is then realised via Von Neumann self-adjoint extension theory \cite{Albeverio}. We first show that almost all of the eigenfunctions of this operator equidistribute in position space. We then prove that there is a positive density sequence of eigenfunctions which do not equidistribute in momentum space, specifically we can find a subsequence that partially localises in a given direction.

\begin{figure}[h]
\centering
\subfloat[$\lambda_{14322} \approx$ 203.630]{
  \includegraphics[width=60mm]{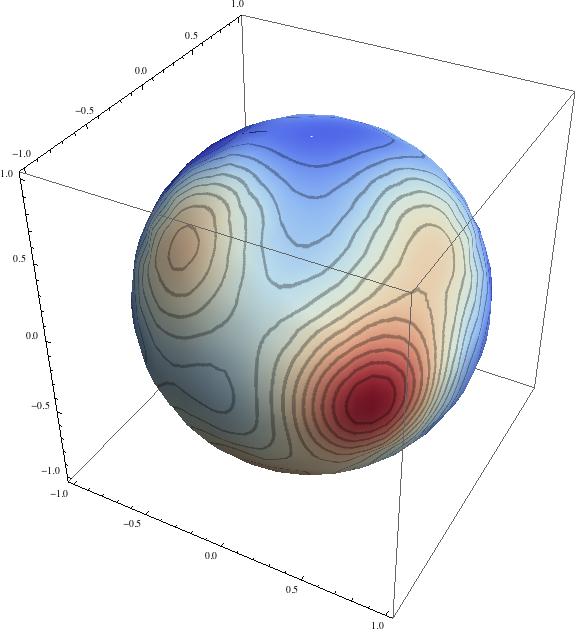}
}
\subfloat[$\lambda_{23985} \approx$ 292.147]{
  \includegraphics[width=60mm]{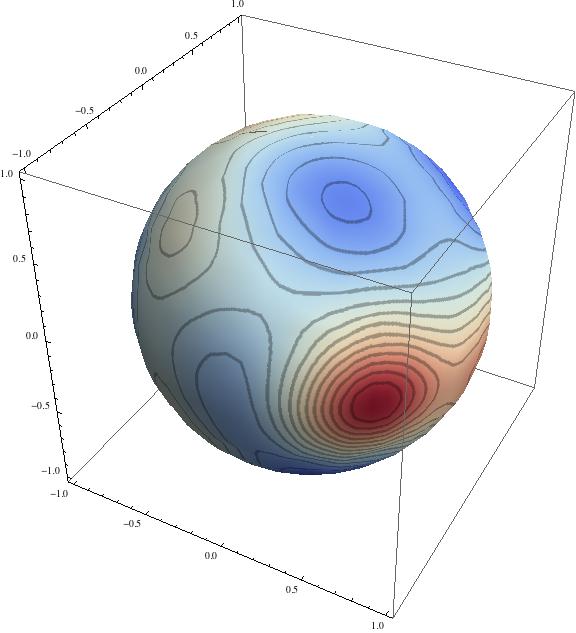}
}
\hspace{0mm}
\subfloat[$\lambda_{45414} \approx$ 454.925]{
  \includegraphics[width=60mm]{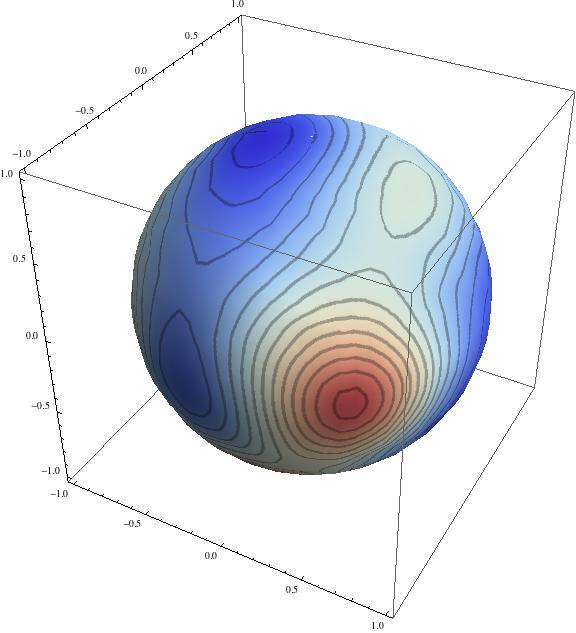}
}
\subfloat[$\lambda_{65109} \approx$ 583.445]{
  \includegraphics[width=60mm]{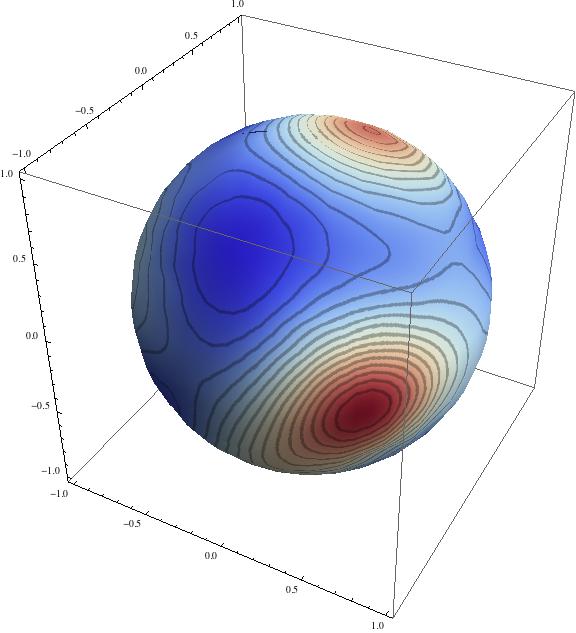}
}
\caption{A collection of non-consecutive eigenfunctions in momentum space with eigenvalue $\lambda$ showing partial localisation in the fixed direction $(1,-1,0)$. We again use fixed quasimomentum $k = (\tfrac{1}{\sqrt{2}}, \tfrac{1}{\sqrt{3}},\tfrac{1}{\sqrt{5}})$.}
\end{figure}

\newpage
\newpage
\section{Setup}
Consider the positive operator $-\Delta_k$ on $\TT^3 = \RR^3/(2\pi\ZZ^3)$ defined by 
\begin{align}\label{operatordef}
\Delta_ k = \left( \frac{\partial}{\partial x} + \i  k_1 \right)^2 + \left( \frac{\partial}{\partial y} + \i k_2 \right)^2 + \left( \frac{\partial}{\partial z} +\i  k_3 \right)^2.
\end{align}
The eigenfunctions of this operator are the complex exponentials
\begin{align}
\frac{1}{(2 \pi)^{3/2}} \e^{\i \langle \xi, \, x \rangle}
\end{align}
with eigenvalue $|\xi+k|^2$, $\xi \in \ZZ^3$. Note here that when the components of $k$ are irrational and linearly independent, then all such eigenvalues are distinct. This will be of major importance in the subsequent proofs. We will write $\scrN=\{n_j \,|\, j\in\NN\}$ to denote the ordered sequence of these eigenvalues. Equivalently, we could consider the standard Laplacian on $\TT^3$ on functions that satisfy the quasiperiodic boundary conditions $\psi(x+\gamma) = \e^{i \langle\gamma,k\rangle} \psi(x)$ for $\gamma$ in $2\pi\ZZ^3$. In this case the eigenfunctions are proportional to the exponentials $\e^{\i \langle \xi + k, x \rangle}$ and again have corresponding eigenvalue $|\xi+k|^2$ - it turns out that the first formulation is more convenient in our case. It is worth noting that this operator occurs naturally when considering the Laplacian on $\RR^3$ with some periodic potential. It is known that provided $V(x+\gamma) = V(x)$ for all $\gamma \in 2\pi\ZZ^3$ then the operator on $\RR^3$ given by $-\Delta + V(x)$ has a direct integral decomposition into operators on $\TT^3$ of the form $-\Delta_k + V(x)$. Full details of this procedure for a general operator can be found in \cite{Cats}. We consider the perturbation of the operator $-\Delta_k$ by a $\delta$ potential at a given point $x_0\in\TT^3$. We realise the perturbed operator 
\begin{align}
H_ k = -\Delta_k + \delta_{x_0}
\end{align}
via self-adjoint extension theory. Details of this calculation can be found in e.g. \cite{Albeverio}. The idea is that if we restrict our operator to functions vanishing at the point $x_0$, it should act like $-\Delta_k$. This operator is then positive symmetric but not self-adjoint, so we extend the domain of functions in such a way that self-adjointness is regained. If we define the restricted Laplacian, $-\Delta_0 := -\Delta \mid_{\scrD_0}$ with 
\begin{align}
\scrD_0 := C^\infty(\TT^3 / \{x_0\}),
\end{align}
then the deficiency indices are $(1,1)$ and the deficiency elements are the Green's functions, $G_{\pm \i}(x,x_0)$, where we define $G_\lambda$ by
\begin{align}
G_\lambda(x, x_0) := (\Delta_ k + \lambda)^{-1} \delta(x-x_0) \overset{\L^2}{=} -\frac{1}{8 \pi^3} \sum_{\xi \in \ZZ^3} \frac{ \e^{\i \langle \xi,\, x - x_0 \rangle}}{|\xi+ k|^2 - \lambda}.
\end{align}
Throughout the paper we will also use $g_{\lambda} = G_\lambda/\|G_{\lambda}\|$ to denote the normalised Green's functions. There therefore exists a 1-parameter family of self-adjoint extensions parametrised by $\phi$ which we denote by $\Delta_{k,\phi}$. The domains of these operators consist of functions $f$ such that
\begin{align}
f(x) = C \big( \cos(\phi/2) \frac{1}{4 \pi |x - x_0|} + \sin(\phi/2) \big) + o(1)
\end{align}
as $x \to x_0$. The domain of $\Delta_{k,\phi}$ can be written
\begin{align}
\scrD_\phi &= \Big\{h + c G_\i( \cdot, x_0) + c \e^{\i \phi} G_{-\i}(\cdot, x_0) \mid h \in \scrD_0, c \in \CC, \phi \in (-\pi,\pi) \Big\},
\end{align}
and the action of $\Delta_{k,\phi}$ is given by
\begin{align}
-\Delta_{k,\phi} f &= -\Delta_k \, h + c \i G_\i ( \cdot , x_0) - c \e^{\i \phi} \i G_{-\i}( \cdot, x_0).
\end{align}
The new perturbed eigenvalues are given by solutions of the equation
\begin{align}
\sum_{\xi \in \ZZ^3} \left( \frac{1}{|\xi+ k|^2-\lambda} - \frac{|\xi+ k|^2}{|\xi+ k|^4 + 1} \right) = c_0 \tan(\phi/2),
\end{align}
where
\begin{align}
c_0 &= \sum_{\xi \in \ZZ^3} \frac{1}{|\xi+ k|^4+1}.
\end{align}
The set of perturbed eigenvalues will be denoted by $\Lambda$.
\section{Statement of Results}
We state the main results as two separate theorems, the first concerning pure position observables, the second concerning full phase space observables. To deal with phase space we first need to define quantisation. We follow the approach used in \cite{Yesha_QE}. Consider a classical symbol $a \in C^\infty(S^*\TT^3)$, where $S^*\TT^3 \simeq \TT^3 \times S^2$. We define the quantisation $\Op(a)$ by
\begin{align}
(\Op(a) f)(x) = \sum_{\xi \in \ZZ^3} \e^{\i \langle \xi,x \rangle} a(x,\overline{\xi+ k}) \hat{f}(\xi),
\end{align}
where we use the notation $\overline{\xi} = \tfrac{\xi}{|\xi|}$  and $|\xi+k|$ is assumed to be nonzero for all $\xi$.
We can then expand $a$ in functions $\e_{\zeta,l,m}(x,\xi) = Y_{l,m}(\overline{\xi}) \e^{\i \langle \zeta , x \rangle}$, where $Y_{l,m}(\xi)$ is the (normalised) spherical harmonic of degree $l$ and order $m$. Specifically we consider some finite polynomial $P$ defined by
\begin{align}
P(x,\xi) = \sum_{|\zeta|\leq N_1} \sum_{l \leq N_2} \sum_{|m|\leq l} c_{\zeta,l,m}\e_{\zeta,l,m}(x,\xi),
\end{align}
and claim that for all $a \in C^{\infty}(S^*\TT^3)$ there exist $N_1$ and $N_2$ such that for all $(x,\xi) \in S^*\TT^3$ and multi-indices $\alpha$ with $|\alpha| < 2$ we have
\begin{align}
| \partial_x^\alpha(a(x,\xi) - P(x,\xi))| < \epsilon.
\end{align}
In light of this it suffices to prove our theorem only for these finite polynomials - the extension to a wider class of functions can be performed by expanding the function in a basis of these polynomials, truncating at some finite order, and controlling the error term  (see \cite{Yesha_QE} for details). We are now able to state the main results. Let $\Lambda$ denote the sequence of perturbed eigenvalues.
\begin{defn} A vector $k \in \RR^d$ is said to be Diophantine of type $\kappa$ if there exists a constant $C$ such that for all $m \in \RR^d, q \in \NN$ we have
\begin{align}
\max_j \left| k_j - \frac{m_j}{q} \right| > \frac{C}{q^\kappa}.
\end{align}
The smallest possible value of $\kappa$ is  $\kappa = 1+\tfrac{1}{d}$. In this case $k$ is called badly approximable. We now consider exclusively $3$ dimensions, where $\kappa \geq 4/3$.
\end{defn}
The first theorem concerns position space equidistribution and is proved in Section 5.
\begin{thm} Fix $\phi \in (-\pi,\pi)$. Assume the components of $(1,k)$ are linearly independent over $\QQ$. Then, there is a density one subset $\Lambda' \subset \Lambda$ such that for all observables $a \in C^{\infty}(\TT^3)$  we have
\begin{align}
\lim_{\lambda \to \infty} \langle a(x) g_\lambda(x), g_\lambda(x) \rangle = \frac{1}{8\pi^3} \int_{\TT^3} a(x) \, dx
\end{align}
with $\lambda \in \Lambda'$.
\label{theorem1}
\end{thm}
The second theorem concerns simultaneous equidistribution in position space and partial localisation in momentum space and is proved in Section 6.
\begin{thm} Fix $\phi \in (-\pi,\pi)$. Let $k$ be diophantine of type $\kappa \in [4/3, 2)$ and assume the components of $(1,k)$ are linearly independent over $\QQ$. Then, for all $\epsilon > 0$ there is a subset $M_\epsilon \subset \NN$ of density at least $1-\epsilon$ such that for all subsequences $(\lambda_n)_{n \in M_\epsilon} $, there exists a further subsequence $(\lambda_{n_j})_{j\in\NN}$ such that for all observables $a \in C^\infty(S^*\TT^3)$ we have
\begin{align} \label{theorem2equation}
\lim_{j\to\infty} \langle \Op(a(x,\xi)) g_{\lambda_{n_j}}(x),g_{\lambda_{n_j}}(x) \rangle = \frac{1}{\vol(S^*\TT^3)} \int_{S^*\TT^3} a(x,\xi) \ dx \, d\mu(\bar{\xi})
\end{align}
where $\mu$ has a positive proportion of its mass supported on a finite number of points.\label{theorem2}
\end{thm}

\section{Truncation}
In order to consider only finite sums we define a truncated Green's function. Define $A(\lambda,L)$ by
\begin{align}
A(\lambda,L) = \{ \xi \in \ZZ^3 : | |\xi+ k|^2-\lambda| < L \},
\end{align}
we then define the truncated Green's function by
\begin{align}
G_{\lambda,L}(x, x_0) = -\frac{1}{8 \pi^3} \sum_{\xi \in A(\lambda,L)} \frac{ \e^{\i \langle \xi,\, x - x_0 \rangle}}{|\xi+ k|^2 - \lambda},
\end{align}
and as before we denote by $g_{\lambda,L} = G_{\lambda,L}/ \|G_{\lambda,L}\|$ the corresponding normalised truncated Green's function. We want to show that for $L= \lambda^{-\delta}$ for some $\delta$ this truncation is a good approximation for large $\lambda$. We first need a lower bound on the full Green's function. Define
\begin{align}
\scrN(x) = \{ n \in \scrN \, \mid \, n \leq x\}.
\end{align}
If the components of $(1,k)$ are linearly independent over $\QQ$ then we know the asymptotic behaviour of $\scrN(x)$ to be
\begin{align}
N(x) = \#\scrN(x) = \frac{4}{3} \pi x^{3/2} + O(x^\theta).
\label{weyl}
\end{align}
It is conjectured that $\theta = \tfrac{1}{2} +\epsilon$ for all $\epsilon$, and for $ k=0$, when counting with multiplicities, the current best explicit bound due to Heath-Brown \cite{Heath-Brown} gives $\theta = \tfrac{21}{32}+ \epsilon$ for all $\epsilon > 0$. For our purposes it is required that $\theta < 1$, in fact we will show in the Appendix that we have $\theta < \frac{3}{4} + \epsilon$ independent of $k$.
\begin{lem} \label{4.1}
Let the components of $(1,k)$ be linearly independent over $\QQ$. Then, there is a density one subset of eigenvalues $\Lambda' \subset \Lambda$ such that for $\lambda \in \Lambda'$
\begin{align}
\|G_\lambda\| \gg \lambda^{1/2-\epsilon}.
\end{align}
\end{lem}
\begin{proof}
We have from \eqref{weyl} that
\begin{align}
\frac{1}{N(x)} \sum_{n_k\leq x} (n_k-n_{k-1}) \asymp \frac{x}{N(x)} \asymp x^{-1/2}.
\end{align}
Thus, since $n_k-n_{k-1} > 0$ we must have that for a subsequence of density one that
\begin{align}
n_{k+1}-n_k \ll n_{k+1}^{-1/2+\epsilon}.
\end{align}
Let $n_{k+1} > \lambda > n_k$ and we see
\begin{align}
\|G_\lambda\|^2 \gg \sum_{n \in \scrN} \frac{1}{(n-\lambda)^2} > \frac{1}{(n_{k+1} - \lambda)^2} > \frac{1}{(n_{k+1}-n_k)^2} \gg n_{k+1}^{1-\epsilon} > \lambda^{1-\epsilon}.
\end{align}
\end{proof}

\begin{lem}
Let $L= \lambda^{-\delta}$, then $\|g_{\lambda,L}-g_{\lambda} \| \to 0$ as $\lambda \to \infty$ with $\lambda \in \Lambda'$.
\end{lem}
\begin{proof}
First we see that 
\begin{align}
\|g_{\lambda,L}-g_{\lambda}\| &= \left\| \frac{G_\lambda}{\|G_\lambda\|} - \frac{G_{\lambda,L}}{\|G_{\lambda,L}\|} \right\| \\
&= \left\| \frac{G_\lambda}{\|G_\lambda\|} - \frac{G_{\lambda,L}}{\|G_\lambda\|} + \frac{G_{\lambda,L}}{\|G_\lambda\|} - \frac{G_{\lambda,L}}{\|G_{\lambda,L}\|} \right\| \\
&\leq \frac{\|G_\lambda - G_{\lambda,L}\|}{\|G_\lambda\|} + \|G_{\lambda,L}\| \left| \frac{1}{\|G_\lambda\|} - \frac{1}{\|G_{\lambda,L}\|} \right| \\
&\leq 2 \frac{\|G_\lambda - G_{\lambda,L}\|}{\|G_\lambda\|}.
\end{align}
Then we have
\begin{align}
\| G_\lambda - G_{\lambda,L} \|^2 \ll  \sum_{||\xi+ k|^2 - \lambda|>L} \frac{1}{(|\xi+ k|^2 - \lambda)^2}.
\end{align}
We evaluate the lattice sum via Abel summation, which tells us that for a smooth function $f$ we have
\begin{align}
\sum_{n_A<|\xi+ k|^2<n_B} f(|\xi+ k|^2) = N(n_B)f(n_B)-N(n_A)f(n_{A+1}) - \int_{n_{A+1}}^{n_B} f'(t)N(t)\, dt.
\end{align}
Integrating by parts we see
\begin{align}
\sum_{n_A<|\xi+ k|^2<n_B} f(|\xi+ k|^2) &= 2 \pi \int_{n_{A+1}}^{n_B} f(t) t^{1/2} \, dt \\ &+ O (n_B^{\theta}f(n_B) - n_A^{\theta}f(n_{A+1}) ) + O(\int_{n_{A+1}}^{n_B} |f'(t)|t^{\theta} \, dt).
\end{align}
Applying this to $f(n) = \frac{1}{(n-\lambda)^2}$ with $n_A = n_0$ and  $n_B < \lambda - L < n_{B+1}$ we see
\begin{align}
\sum_{n < \lambda - L} \frac{1}{(n - \lambda)^2} = 2 \pi \int_{n_1}^{n_B} \frac{n^{1/2}}{(n-\lambda)^2} \, dn + O\left(\frac{n_B^\theta}{(n_B-\lambda)^2} \right) + O\left( \int_{n_1}^{n_B} \frac{n^\theta}{(\lambda-n)^3} \, dn\right).
\end{align}
We can bound the integral by
\begin{align}
\int_{n_1}^{n_B} \frac{n^{1/2}}{(n-\lambda)^2} \, dn &\leq \lambda^{1/2}\int_{n_1}^{n_B} \frac{1}{(n-\lambda)^2} \, dn \\
&\leq \frac{\lambda^{1/2}}{L} \leq \frac{\lambda^{\theta}}{L^2}.
\end{align}
Similarly we see
\begin{align}
\frac{n_B^\theta}{(n_B-\lambda)^2} \leq \frac{\lambda^{\theta}}{L^2},
\end{align}
and also
\begin{align}
\int_{n_1}^{n_B} \frac{n^\theta}{(\lambda-n)^3} \ll \frac{\lambda^{\theta}}{L^2}.
\end{align}
Now repeating this procedure with $n_A < \lambda+ L < n_{A+1}$ and $n_B = \infty$ we obtain
\begin{align}
\sum_{n>\lambda-L} \frac{1}{(n-\lambda)^2} = 2 \pi \int_{n_{A+1}}^\infty \frac{n^{1/2}}{(n-\lambda)^2} \, dn + O\left( \frac{n_A^\theta}{(n_A-\lambda)^2} \right) + O\left( \int_{n_A}^{\infty} \frac{n^\theta}{(\lambda-n)^3} \, dn\right).
\end{align}
For the first integral we write
\begin{align}
\int_{n_{A+1}}^\infty \frac{n^{1/2}}{(n-\lambda)^2} \, dn &= \int_{n_{A+1}-\lambda}^\infty \frac{(s+\lambda)^{1/2}}{s^2} \, ds \nonumber \\
&\leq \int_L^\lambda \frac{(s+\lambda)^{1/2}}{s^2} \, ds + \int_\lambda^\infty \frac{(s+\lambda)^{1/2}}{s^2} \, ds \nonumber \\
&\ll \frac{\lambda^{1/2}}{L} \ll \frac{\lambda^{\theta}}{L^2}. 
\end{align}
For the second term we have immediately
\begin{align}
\frac{n_A^\theta}{(n_A-\lambda)^2} \ll \frac{\lambda^\theta}{L^2}.
\end{align}
For the third term we see
\begin{align}
\int_{n_A}^{\infty} \frac{n^\theta}{(\lambda-n)^3} \, dn &= \int_{n_{A+1}-\lambda}^\infty \frac{(s+\lambda)^{\theta}}{s^3} \, ds \nonumber \\
&\ll \int_{n_{A+1}-\lambda}^\infty \frac{s+\lambda}{s^3} \, ds  \ll \frac{1}{\lambda}.
\end{align}
Putting all of this together we see
\begin{align}
\| G_\lambda - G_{\lambda,L} \|^2 \ll \frac{\lambda^{\theta}}{L^2},
\end{align}
and hence, using Lemma \ref{4.1}, that for the normalised Green's functions
\begin{align}
\|g_{\lambda,L}-g_{\lambda}\| \ll \frac{\lambda^{-\tfrac{1-\theta}{2}}}{L} = \lambda^{- (1-\theta)/2 + \epsilon + \delta}
\end{align}
which tends to 0 for all $\delta < \frac{1-\theta}{2} - \epsilon$.
\end{proof}
\begin{cor} Define $g_{\lambda,L}$ as above with $L = \lambda^{-\delta}$ and $0< \delta < \tfrac{1-\theta}{2}-\epsilon$ then
\begin{align}
|\langle \Op(\e_{\zeta,l,m})g_{\lambda,L},g_{\lambda,L} \rangle - \langle \Op(\e_{\zeta,l,m})g_\lambda,g_\lambda \rangle | \to 0.
\end{align}
\label{corollary}
\end{cor}
\begin{proof}
We have
\begin{align}
&|\langle \Op(\e_{\zeta,l,m})g_{\lambda,L},g_{\lambda,L} \rangle - \langle \Op(\e_{\zeta,l,m})g_\lambda,g_\lambda \rangle |  \\
&\leq |\langle \Op(\e_{\zeta,l,m})g_{\lambda,L},g_{\lambda,L} - g_\lambda \rangle| + |\langle \Op(\e_{\zeta,l,m})(g_\lambda-g_{\lambda,L}),g_\lambda \rangle |.
\end{align}
Taking each term and using Cauchy-Schwarz gives
\begin{align}
|\langle \Op(\e_{\zeta,l,m})g_{\lambda,L},g_{\lambda,L} \rangle - \langle \Op(\e_{\zeta,l,m})g_\lambda,g_\lambda \rangle | \leq \|\Op(\e_{\zeta,l,m})\| \|g_{\lambda}-g_{\lambda,L}\| \to 0.
\end{align}
\end{proof}
\section{Equidistribution in Position Space}
The following proposition is key to the proof.
\begin{prop}
Fix $\zeta \neq 0$, $l \in \NN$ and $|m| \leq l$. Let $L = \lambda^{-\delta}$ for some $\delta>0$. Let the components of $(1,k)$ be linearly independent over $\QQ$.  Then, for $\lambda$ sufficiently large we have
\begin{align}
\langle \Op(\e_{\zeta,l,m}) g_{\lambda,L}, g_{\lambda,L} \rangle = 0.
\end{align}
\label{prop1}
\end{prop}
In order to prove this we first need a lemma.
\begin{lem}
Let the components of $(k,1)$ be linearly independent over $\QQ$ and fix $\zeta \in \ZZ^3$ nonzero. Then, there exists some $\epsilon>0$ such that for all $\xi \in \ZZ^3$ we have
\begin{align}
|2\langle \xi+ k,\zeta\rangle + |\zeta|^2| > \epsilon.
\end{align}
\label{nonorthogonal}
\end{lem}
\begin{proof} We have that $$|2\langle \xi+ k,\zeta\rangle + |\zeta|^2| > \| 2 \langle k, \zeta \rangle \|$$
where $\| \cdot \|$ represents the distance to the nearest integer. Since we assumed the components of $(k,1)$ were linearly independent, this is bounded away from zero. Now choose $\epsilon = \|\langle k,\zeta \rangle\|$.
\end{proof}
\begin{proof}[Proof of Proposition \ref{prop1}]
First write
\begin{align}
&|\langle \Op(\e_{\zeta,l,m}) G_{\lambda,L},G_{\lambda,L} \rangle| \nonumber\\
& = \frac{1}{64 \pi^6 \|G_\lambda\|^2} \Big|\langle \sum_{\xi \in A(\lambda,L)} \frac{\e^{\i \langle \xi, x-x_0 \rangle}}{|\xi+ k|^2-\lambda} \e^{\i \langle \zeta, x \rangle} Y_{l,m}(\overline{\xi+ k}), \sum_{\eta\in A(\lambda ,L)} \frac{\e^{\i \langle \eta,x-x_0\rangle}}{|\eta+ k|^2-\lambda}   \rangle \Big| \\ \label{integral}
& = \frac{1}{64 \pi^6 \|G_\lambda\|^2} \Big| \int_{\TT^3}\sum_{\xi,\eta \in A(\lambda,L)} \frac{ \e^{\i \langle \eta-\xi,x-x_0\rangle}}{(|\xi+ k|^2-\lambda)(|\eta+ k|^2-\lambda)} \e^{-\i\langle\zeta,x\rangle} Y_{l,m}^*(\overline{\xi+ k}) \, dx \Big|.
\end{align}
Integrating over $x$ leaves only the terms where $\eta = \xi+\zeta$. However, note that by Lemma \ref{nonorthogonal}, for $\xi \in A(\lambda,L)$,
\begin{align}
||\xi+\zeta+ k|^2 - \lambda| = ||\xi+ k|^2-\lambda + 2 \langle \xi+ k, \zeta \rangle + |\zeta|^2| \gg \epsilon
\end{align}
so $\xi+\zeta \notin A(\lambda,L)$ for $\lambda$ sufficiently large. Thus the integral in \eqref{integral} vanishes. \end{proof}
We are now able to show equidistribution for position space observables.
\begin{proof}[Proof of Theorem \ref{theorem1}] Let $\lambda \in \Lambda'$, and let $a \in C^\infty(\TT^3)$. The operator $\Op(a)$ is then just given by multiplication by $a$. We consider $a$ to be some finite polynomial
\begin{align}
a(x) = \sum_{|\zeta| < N} \hat{a}(\zeta) \, \e^{\i \langle \zeta , x \rangle},
\end{align}
and see from Proposition $\ref{prop1}$ that
\begin{align}
\langle a(x) g_{\lambda,L}, g_{\lambda,L} \rangle &\to \langle \hat{a}(0) g_{\lambda,L},g_{\lambda,L} \rangle \\
&= \left(\int_{\TT^3} a(y) \, \frac{dy}{8\pi^3} \right) \left( \int_{\TT^3} |g_{\lambda,L}(x)|^2 \, \frac{dx}{8\pi^3} \right) \nonumber \\
& = \int_{\TT^3} a(y) \, \frac{dy}{8\pi^3} .\nonumber
\end{align}
The result then follows from Corollary $\ref{corollary}$.
\end{proof}
\section{Localisation in Momentum Space}
Throughout this section we will assume $k$ is diophantine of type $\kappa < 2$. Let $a$ be defined by
\begin{align}a(x,\xi) = \sum_{|\zeta|\leq N_1,l \leq N_2,|m| \leq l} \hat{a}(\zeta,l,m) \e_{\zeta,l,m}(x,\xi)
\end{align}
where $\hat{a}(\zeta,l,m)$ is given by
\begin{align}
\hat{a}(\zeta,l,m) =\frac{1}{8\pi^3} \int_{S^2}\int_{\TT^3} a(x,\xi) \e^{-\i \langle x, \zeta \rangle} Y^*_{l,m}(\overline{\xi}) \, dx d\sigma(\xi).
\end{align}
We thus have that
\begin{align}
\langle \Op(a) g_{\lambda,L},g_{\lambda,L} \rangle &\sim \langle \sum_{l,m} \hat{a}(0,l,m) \Op(\e_{0,l,m}) g_{\lambda,L},g_{\lambda,L} \rangle \\
& = \|G_\lambda\|^{-2}  \frac{1}{64 \pi^6} \sum_{l,m} \sum_{\xi \in A(\lambda,L)} \hat{a}(0,l,m)\frac{Y_{l,m} (\overline{\xi+ k})}{(|\xi+ k|^2-\lambda)^2}  \nonumber\\
&\asymp \|G_\lambda\|^{-2}\sum_{\xi \in A(\lambda,L)} \int_{S^*\TT^3} a(x,\eta)  \, dx \, \frac{ \, \delta({\bar{\eta}-\overline{\xi+k}})}{(|\xi+ k|^2-\lambda)^2} d\sigma(\eta). \nonumber
\end{align}
Thus the component of the spectral measure for each fixed $|\xi+ k|^2 = m$ on $\TT^3 \times S^2$ consists of $\operatorname{Leb} \times \delta_{\overline{\xi+ k}}$. The full (unnormalised) spectral measure is thus a weighted sum of a growing number of $\delta$ masses that become dense on $S^2$. We aim to show that for a positive density subsequence of $\lambda$, the tails of this sum can be bounded uniformly in $\lambda$ such that a positive proportion of its density will be supported on a finite number of points. A key feature of the proof is the work by Marklof \cite{Marklof_IQF2}. We now recapitulate the necessary results in dimension 3.
\begin{defn}
Let $\scrN = \{n_j \mid j \in \NN\}$ be defined as before. Let $\psi_1, \psi_2 \in \scrS(\RR^+)$ be Schwartz functions, and let $h \in C_0(\RR)$ be compactly supported, and $\hat{h}$ its Fourier transform. Define the generalised pair correlation function by
\begin{align}
R(\psi_1,\psi_2,h,T) = \frac{3}{4 \pi T^{3/2}} \sum_{\substack{i,j=1 \\ i \neq j}}^{\infty} \psi_1\left(\frac{n_i}{T} \right) \, \psi_2 \left( \frac{n_j}{T} \right) \, \hat{h}(\sqrt{T} \, (n_i-n_j)).
\end{align}
\begin{thm}[See Theorem 2.5 in \cite{Marklof_IQF2} for the statement in full generality] \label{IQF}
Let $k$ be diophantine of type $\kappa < 2$, and that assume the components of $(k,1)$ are linearly independent over $\QQ$. Then
\begin{align}
\lim_{T \to \infty} R(\psi_1,\psi_2,h,T) = \,  3 \pi \int \hat{h}(s) \, ds \, \int_0^{\infty} \psi_1(r) \, \psi_2(r) \, r \, dr. \nonumber
\end{align}
\end{thm}
We now proceed with the proof.
\end{defn}
\begin{lem}
For $G \geq 1$, we have that $\# \{n_i \in \scrN(T) : n_{i+1} - n_i > G/\sqrt{n_{i+1}} \} < T^{3/2}/G$.
\label{lem1}
\end{lem}
\begin{proof}
We see that
\begin{align}
\sum_{n_i \leq T} \sqrt{n_{i+1}} (n_{i+1} - n_i) &< \sum_{n_i \leq T} (n_{i+1}^{3/2} - n_i^{3/2}) \\
&< T^{3/2}.
\end{align}
Thus by Chebyshev's inequality we see
\begin{align} 
\# \{n_i \leq T : s_i = n_{i+1} - n_i > G/\sqrt{n_{i+1}} \} < T^{3/2}/G.
\end{align}
\end{proof}
\begin{lem}
Given $D>0$, $E \geq 1$,
\begin{align}
\#\{ n \in \scrN(T) : |\scrN(T) \cap [n- \tfrac{D}{\sqrt{n}},n+ \tfrac{D}{\sqrt{n}}]| > E+1 \} \ll \frac{DT^{3/2}}{E}.
\end{align}
\label{lem2}
\end{lem}
\begin{proof}
We have that
\begin{align}
\sum_{n\in\scrN(T)}& (|\scrN(T)\cap [n- \tfrac{D}{\sqrt{n}},n+ \tfrac{D}{\sqrt{n}}]|-1) \\
&= \#\{n,m \in \scrN(T) : m\neq n, \sqrt{n}|n-m| \leq D \} \\[0.1cm]
& \ll \#\{n,m \in \scrN(T)\backslash \scrN(T/2) : m\neq n, \sqrt{T}|n-m| \leq D \}.
\end{align}
Since we assumed $k$ was diophantine, we may apply Theorem \ref{IQF} with $\psi_1 = \psi_2$ the indicator function of $[1/2,1]$, and $\hat{h}$ the indicator function of $[-D,D]$. This gives us the asymptotics
\begin{align}
 \#\{n,m \in \scrN(T)\backslash \scrN(T/2) : m\neq n, \sqrt{T}|n-m| \leq D \} \sim 3 \pi^2 D T^{3/2}.
\end{align}
Again by Chebyshev's inequality we conclude
\begin{align}
\#\{ n \in \scrN(T) : |\scrN(T) \cap [n- \tfrac{D}{\sqrt{n}},n+ \tfrac{D}{\sqrt{n}}] > E+1 \} \ll \frac{DT^{3/2}}{E}.
\end{align}
\end{proof}
\begin{lem}
For all $A>1$
\begin{align}
\label{lem3sum}
\sum_{\substack{n,m \in \scrN(T) \\ \sqrt{m}|n-m|>A}} \frac{1}{m(n-m)^2} \ll \frac{T^{3/2}}{A^{1/3}}.
\end{align}
\label{lem3}
\end{lem}
\begin{proof}
We first define
\begin{align}
M(k) := |\{n\in\scrN : n^{3/2} \in [k,k+1] \}|.
\end{align}
Then we deduce an $L^2$ bound on $M(k)$ by
\begin{align}
\sum_{k \leq T^{3/2}} M(k)^2 &= \sum_{k\leq T^{3/2}} |\{m,n \in \scrN : m^{3/2},n^{3/2} \in [k,k+1]\}| \\
&\leq |\{m,n \in \scrN : m^{3/2},n^{3/2} \leq T^{3/2}+1, m^{3/2}-n^{3/2} \in [-1,1] \}|
\end{align}
which again by Theorem \ref{IQF} gives us
\begin{align}
\sum_{k \leq T^{3/2}} M(k)^2 \ll T^{3/2}.
\end{align}
Note that we can write
\begin{align}
\sqrt{m}|n-m| = \frac{\sqrt{m}}{\sqrt{m} + \sqrt{n}} (\sqrt{m}|n-m| + \sqrt{n}|n-m|) \geq \frac{\sqrt{m}}{\sqrt{m} + \sqrt{n}} |n^{3/2}-m^{3/2}|, \nonumber
\end{align}
and also that $\sqrt{m}|n-m| < |n^{3/2}-m^{3/2}|$. Hence, we can bound the sum in \eqref{lem3sum} by
\begin{align}
\sum_{\substack{n,m \in \scrN(T) \\ \sqrt{m}|n-m|>A}} \frac{1}{m(n-m)^2} &\ll \sum_{\substack{n,m \in \scrN(T) \\ |n^{3/2}-m^{3/2}|>A }} \frac{(1+ \sqrt{\tfrac{n}{m}})^2}{(n^{3/2}-m^{3/2})^2} \\
& = \sum_{k=A}^{\lfloor T^{3/2} \rfloor} \sum_{\substack{n,m \in \scrN(T) \\ |n^{3/2}-m^{3/2}| \in [k,k+1] }} \frac{(1+ \sqrt{\tfrac{n}{m}})^2}{(n^{3/2}-m^{3/2})^2} \nonumber.
\end{align}
Now, when $m>n$ we can immediately conclude
\begin{align}
\sum_{k=A}^{\lfloor T^{3/2} \rfloor} \sum_{\substack{n,m \in \scrN(T) \\ |n^{3/2}-m^{3/2}| \in [k,k+1] }} \frac{(1+ \sqrt{\tfrac{n}{m}})^2}{(n^{3/2}-m^{3/2})^2} & < 4 \sum_{k=A}^{\lfloor T^{3/2} \rfloor} \frac{1}{k^2} |\{m,n \in \scrN(T) :  (n^{3/2}-m^{3/2}) \in [k,k+1] \}|  \nonumber\\
&\leq 4 \sum_{k=A}^{\lfloor T^{3/2} \rfloor} \frac{1}{k^2} \sum_{m\leq T^{3/2}} M(m)(M(m+k)+M(m+k+1)). \end{align}
By Cauchy-Schwarz we may bound this above by
\begin{align}
T^{3/2} \sum_{k=A}^{\lfloor T^{3/2} \rfloor} \frac{1}{k^2} \ll \frac{T^{3/2}}{A}.  \nonumber
\end{align}
When $m<n$, we see that
\begin{align}
n^{3/2} - m^{3/2} \in [k,k+1] \implies \left(\frac{n}{m}\right)^{3/2} \leq 1 + \frac{k+1}{m^{3/2}}.
\end{align}
We know that $m$ is bounded away from zero, say $m>C$, then we must have 
\begin{align}
\left(\frac{n}{m}\right)^{1/2} \leq C^{-1/2} (k+1+C^{3/2})^{1/3}.
\end{align}
Repeating the previous argument in this regime yields
\begin{align}
\sum_{k=A}^{\lfloor T^{3/2} \rfloor} \sum_{\substack{n,m \in \scrN(T) \\ |n^{3/2}-m^{3/2}| \in [k,k+1] }} \frac{(1+ \sqrt{\tfrac{n}{m}})^2}{(n^{3/2}-m^{3/2})^2} & \ll \sum_{k=A}^{\lfloor T^{3/2} \rfloor} \frac{1}{k^{4/3}} |\{m,n \in \scrN(T) :  (n^{3/2}-m^{3/2}) \in [k,k+1] \}|  \nonumber\\
&\leq \sum_{k=A}^{\lfloor T^{3/2} \rfloor} \frac{1}{k^{4/3}} \sum_{m\leq T^{3/2}} M(m)(M(m+k)+M(m+k+1)) \\
& \ll  T^{3/2} \sum_{k=A}^{\lfloor T^{3/2} \rfloor} \frac{1}{k^{4/3}} \ll \frac{T^{3/2}}{A^{1/3}},
\end{align}
where the final line follows from Cauchy-Schwarz as before.
\end{proof}
We are now ready to prove the second main theorem.
\begin{proof}[Proof of Theorem \ref{theorem2}] Define $\scrN'$ as follows, first remove from $\scrN$ all points $m$ whose nearest left neighbour is further than $G/\sqrt{m}$, by Lemma \ref{lem1} we are left with a subsequence of density at least $1-1/G$. Now choose $D$ and fix $E$ large enough such that 
\begin{align}
|\{ m \in \scrN(T) : |\scrN(T) \cap [m- \tfrac{D}{\sqrt{m}},m+ \tfrac{D}{\sqrt{m}}] > E+1 \}| \leq \frac{T^{3/2}}{G}
\end{align}
which is possible by Lemma \ref{lem2}. Removing these points leaves us with a subsequence of density at least $1-2/G$. Finally, by Lemma \ref{lem3}, and Chebyshev's inequality we choose $F$ large enough such that 
\begin{align}
|\{m \in \scrN(T) : \sum_{\substack{n \in \scrN(T) \\ \sqrt{m}|n-m| > D}} \frac{1}{(n-m)^2} > F \, m \} | \leq \frac{T^{3/2}}{G}.
\end{align}
Removing these points leaves us with a subsequence of density at least $1-3/G$. Thus if we consider pure momentum observables and for $m\in\scrN'$ denote by $\mu_m$ the delta measure on the point corresponding to the direction $\xi+ k$ with $|\xi+ k|^2 = m$, we see that the unnormalised measure associated to $G_{\lambda_m}$ is
\begin{align}
\sum_{n \in \scrN} \frac{\mu_n}{(n-\lambda_m)^2} = \frac{\mu_m}{(m-\lambda_m)^2} + \sum_{\substack{n\in\scrN \\ 0 < |n-m| < \tfrac{D}{\sqrt{m}}}} \frac{\mu_n}{(n-\lambda_m)^2} + \sum_{\substack{n \in \scrN \\ |n-m| > \tfrac{D}{\sqrt{m}}}} \frac{\mu_n}{(n-\lambda_m)^2}.
\end{align}
We know that the mass of the first term is $\gg m/G^2$, the mass of the second sum has at most $E$ terms, and the mass of the third is bounded above by $F \, m$. Thus the normalised measure will have a positive proportion of its mass on a finite number of points. The theorem then follows from compactness of $S^*\TT^3$ and by setting $\epsilon = 3/G$ and defining $M_\epsilon$ by $j \in M_\epsilon \iff  n_j \in \scrN'$ where $n_j$ is the $j^{th}$ ordered unperturbed  eigenvalue.
\end{proof}

\appendix
\section{}
\begin{prop}
Let $S(R) = \# \{ |\xi+k| < R \mid \xi \in \ZZ^3\} $ denote the number of shifted lattice points inside a ball of radius $R$. Then we have that
\begin{align}
S(R) = \frac{4}{3} \pi R^3 + O(R^{3/2 +\epsilon}).
\end{align}
\end{prop}
\begin{proof}
We bound the quantity $S(R)$ above and below by sums over the indicator function of a shifted ball convolved with some smooth bump function with smoothing parameter $\delta$. We can then employ Poisson summation and tune $\delta$ in such a way that the error terms vanish.  Let $B_k(R)$ denote the ball of radius $R$ centred at $k$, and write $\psi_\delta(x) = \delta^{-3} \psi(x/\delta)$ where $\psi$ is some smooth function with compact support in $B_0(1)$ normalised such that $\hat{\psi}(0) = 1$. Define $S_\delta(R)$ to be the smoothed sum
\begin{align}
S_\delta(R) = \sum_{x \in \ZZ^3} \chi_{B_k(R)}* \psi_\delta(x).
\end{align}
Note that we have $S(R-\delta) \leq S_\delta(R) \leq S(R+\delta)$. By Poisson summation we see
\begin{align}
\sum_{x \in \ZZ^3} \chi_{B_k(R)} * \psi_\delta(x) = \sum_{\xi \in \ZZ^3} \hat{\chi}_{B_k(R)}(\xi) \hat{\psi}_\delta(\xi).
\end{align}
Computing the term $\xi = 0$ yields
\begin{align}
\int_{\RR^3} \chi_{B_k(R)}(x) \, dx &= \frac{4}{3} \pi R^3.
\end{align}
For $\xi \neq 0$, the Fourier coefficients $\hat{\chi}_{B_k(R)}(\xi)$ are given by
\begin{align}
\hat{\chi}_{B_k(R)}(\xi) = \e^{-2 \pi \i \langle k, \xi \rangle} \frac{1}{2 \pi^2 |\xi|^3} \left( \sin (2\pi R |\xi|) - 2 \pi R |\xi| \cos( 2 \pi R |\xi|) \right).
\end{align}
We also have that
\begin{align}
\int_{\RR^3} \delta^{-3} \psi(x/\delta) \e^{-2 \pi \i \langle x, \xi \rangle}  \, dx
&= \int_{\RR^3} \delta^{-3} \psi(x/\delta) (4 \pi^2 |\xi|^2)^{-1} (-\Delta) \e^{-2\pi \i \langle x, \xi \rangle} \, dx \nonumber \\
& = (4 \pi^2 |\xi|^2)^{-1}  \int_{\RR^3} \delta^{-3}  \e^{-2\pi \i \langle x, \xi \rangle} (-\Delta) \psi(x/\delta) \, dx \\
&= (4 \pi^2 |\xi|^2 \delta^2)^{-1}  \int_{\RR^3} \e^{-2\pi \i \delta \langle y, \xi \rangle} (-\Delta) \psi(y) \, dy. \nonumber
\end{align}
We thus have that $|\hat{\psi}_\delta(\xi)| \ll \min \{1, (\delta |\xi|)^{-2} \} \leq (\delta |\xi|)^{-(1+\epsilon)}$. Plugging these asymptotics in to the sum gives
\begin{align}
\sum_{x \in \ZZ^3 \backslash \{0\}} \chi_{B_k(R)}* \psi_\delta(x)
& \ll R \delta^{-(1+\epsilon)} \sum_{\xi \in \ZZ^3 \backslash \{0\}} |\xi|^{- (3+\epsilon)}\\
& \ll R \delta^{-(1+\epsilon)}. \nonumber
\end{align}
We thus have that
\begin{align}
S(R) \leq S_\delta(R+\delta) &= \frac{4}{3} \pi (R+\delta)^3 + O(R \delta^{-(1+\epsilon)}) \\
&= \frac{4}{3} \pi R^3 + O(R^2\delta + R \delta^{-(1+\epsilon)}), \nonumber
\end{align}
and similarly that
\begin{align}
S(R) \geq S_\delta(R-\delta) = \frac{4}{3} \pi R^3 + O(R^2\delta + R \delta^{-(1+\epsilon)}).
\end{align}
Setting $\delta = R^{-1/2}$ yields the result.
\end{proof}
\section*{Acknowledgements}
I would like to thank the Isaac Newton Institute for Mathematical Sciences, Cambridge, for support and hospitality during the programme `Periodic and Ergodic Spectral Problems' where work on this paper was undertaken, specifically the useful comments of Nadav Yesha, Henrik Uebersch\"{a}r, Zeev Rudnick and P\"{a}r Kurlberg. I would also like to thank Jens Marklof for useful advice throughout and Ilya Vinogradov for suggestions concerning the Appendix.
\bibliographystyle{plain}
\bibliography{QP}

\begin{thebibliography}{10}

\bibitem{solvable_models}
S.~Albeverio, F.~Gesztesy, and H.~Holden.
\newblock {Solvable Models in Quantum Mechanics}.
\newblock 1988.

\bibitem{Albeverio}
S.~Albeverio and P.~Kurasov.
\newblock {\em Singular Perturbations of Differential Operators}.
\newblock Cambridge University Press, 2000.
\newblock Cambridge Books Online.

\bibitem{Colin}
Y.~Colin~de Verdi\`{e}re.
\newblock Ergodicit\'{e} et fonctions propres du laplacien.
\newblock {\em Comm. Math. Phys.}, 102(3):497--502, 1985.

\bibitem{Seba2}
P.~Exner and P.~Seba.
\newblock {Point interaction in dimension two and three as model of small
  scatterers}.
\newblock {\em Phys.Lett.}, A222:1, 1996.

\bibitem{Heath-Brown}
D.~R. Heath-Brown.
\newblock Lattice points in the sphere.
\newblock In {\em Number theory in progress, {V}ol. 2
  ({Z}akopane-{K}o\'scielisko, 1997)}, pages 883--892. de Gruyter, Berlin,
  1999.

\bibitem{Kronig_Penney}
R.~de~L. Kronig and W.~G. Penney.
\newblock Quantum mechanics of electrons in crystal lattices.
\newblock {\em Proceedings of the Royal Society of London. Series A, Containing
  Papers of a Mathematical and Physical Character}, 130(814):pp. 499--513,
  1931.

\bibitem{Kurlberg_Rosenzweig}
P.~{Kurlberg} and L.~{Rosenzweig}.
\newblock {Scarred eigenstates for arithmetic toral point scatterers}.
\newblock {\em ArXiv e-prints}, August 2015.

\bibitem{Ueberschaer_Kurlberg_QE}
P.~Kurlberg and H.~Uebersch\"{a}r.
\newblock Quantum ergodicity for point scatterers on arithmetic tori.
\newblock {\em Geometric and Functional Analysis}, 24(5):1565--1590, 2014.

\bibitem{Marklof_IQF2}
Jens Marklof.
\newblock Pair correlation densities of inhomogeneous quadratic forms, ii.
\newblock {\em Duke mathematical journal}, 115(3):409--434, 2002.

\bibitem{Cats}
Leonid Parnovski and Alexander~V. Sobolev.
\newblock Bethe-sommerfeld conjecture for periodic operators with strong
  perturbations.
\newblock {\em Inventiones mathematicae}, 181(3):467--540, 2010.

\bibitem{Rudnick_Uebershaer}
Z.~{Rudnick} and H.~{Uebersch{\"a}r}.
\newblock {Statistics of Wave Functions for a Point Scatterer on the Torus}.
\newblock {\em Communications in Mathematical Physics}, 316:763--782, December
  2012.

\bibitem{Shig1}
T.~Shigehara.
\newblock Conditions for the appearance of wave chaos in quantum singular
  systems with a pointlike scatterer.
\newblock {\em Phys. Rev. E}, 50:4357--4370, Dec 1994.

\bibitem{Shig2}
T.~Shigehara and Taksu Cheon.
\newblock Wave chaos in quantum billiards with small but finite-size scatterer.
\newblock {\em J. Phys. E}, pages 1321--1331, 1996.

\bibitem{Shig3}
T.~Shigehara and Taksu Cheon.
\newblock Spectral properties of three-dimensional quantum billiards with a
  pointlike scatterer.
\newblock {\em Phys. Rev. E}, 1997.

\bibitem{Shnirelman}
Alexander~I Shnirel'man.
\newblock Ergodic properties of eigenfunctions.
\newblock {\em Uspekhi Matematicheskikh Nauk}, 29(6):181--182, 1974.

\bibitem{Seba}
Petr S\ifmmode~\check{e}\else \v{e}\fi{}ba.
\newblock Wave chaos in singular quantum billiard.
\newblock {\em Phys. Rev. Lett.}, 64:1855--1858, Apr 1990.

\bibitem{Ueberschaer_Kurlberg}
H.~{Ueberschaer} and P.~{Kurlberg}.
\newblock {Superscars in the Seba billiard}.
\newblock {\em ArXiv e-prints}, September 2014.

\bibitem{Yesha_EigenfunctionStatistics}
Nadav Yesha.
\newblock Eigenfunction statistics for a point scatterer on a three-dimensional
  torus.
\newblock {\em Annales Henri Poincar\'{e}}, 14(7):1801--1836, 2013.

\bibitem{Yesha_QE}
Nadav Yesha.
\newblock Quantum ergodicity for a point scatterer on the three-dimensional
  torus.
\newblock {\em Annales Henri Poincar\'{e}}, 16(1):1--14, 2015.

\bibitem{Zelditch}
Steven Zelditch.
\newblock Uniform distribution of eigenfunctions on compact hyperbolic
  surfaces.
\newblock {\em Duke Math. J.}, 55(4):919--941, 12 1987.

\end{thebibliography}

\end{document}